\documentclass[12pt]{article}
\usepackage[english]{babel}

\usepackage{amsmath,amsfonts,amssymb,amsthm}
\usepackage{graphicx}
\usepackage{xspace}
\usepackage{a4wide}
\usepackage{hyperref}

\newcommand{\adj}{\mathrm{T}}%{\dagger}

\newtheorem{definition}{Definition}
\newtheorem{theorem}{Theorem}
\newtheorem{example}{Example}
\newtheorem{remark}{Remark}
\newtheorem{lemma}{Lemma}
\newtheorem{corollary}{Corollary}

\newcommand{\bydef}{:=}

\newcommand{\bx}{\mathbf{x}}
\newcommand{\by}{\mathbf{y}}

\newcommand{\be}{\mathbf{e}}
\newcommand{\R}{\mathbb{R}}

\newcommand{\A}{\ensuremath{\mathbb{A}}\xspace}
\newcommand{\B}{\ensuremath{\mathbb{B}}\xspace}
\newcommand{\N}{\mathbb{N}}
\newcommand{\ba}{\mathbf{a}}
\newcommand{\bb}{\mathbf{b}}

\newcommand{\btheta}{\boldsymbol{\theta}}
\newcommand{\bpi}{\boldsymbol{\pi}}

\usepackage[authoryear]{natbib}
\usepackage{authblk}

\begin{document}

\title{Insuperable strategies in two-player and reducible multi-player games}

\author[1,2]{Fabio A. C. C. Chalub\footnote{Corresponding author: facc@fct.unl.pt}}

\author[1,3]{Max O. Souza\footnote{maxsouza@id.uff.br}}

%\author*[1,2]{\fnm{Fabio A. C. C.} \sur{Chalub}}\email{facc@fct.unl.pt}

%\author[1,3]{\fnm{Max O.} \sur{Souza}}\email{maxsouza@id.uff.br}

\affil[1]{Center for Mathematics and Applications (NOVA Math), NOVA School of Science and Technology, Universidade NOVA de Lisboa, Quinta da Torre, 2829-516, Caparica, Portugal.}

%\affil[1]{\orgdiv{Center for Mathematics and Applications (NOVA Math)}, \orgname{NOVA School of Science and Technology, Universidade NOVA de Lisboa}, \orgaddress{\street{Quinta da Torre}, \city{Caparica}, \postcode{2829-516}, %\state{State}, 
%\country{Portugal}}}

\affil[2]{Department of Mathematics, NOVA School of Science and Technology, Universidade NOVA de Lisboa, Quinta da Torre, 2829-516, Caparica, Portugal.}

%\affil[2]{\orgdiv{Department of Mathematics}, \orgname{NOVA School of Science and Technology, Universidade NOVA de Lisboa}, \orgaddress{\street{Quinta da Torre}, \city{Caparica}, \postcode{2829-516}, %\state{State}, 
%\country{Portugal}}}

\affil[3]{Instituto de Matemática e Estatística, Universidade Federal Fluminense, R. Prof. Marcos Waldemar de Freitas Reis, s/n, Bloco H, Niterói, RJ 24120-201, Brasil}

%\affil[3]{\orgdiv{Instituto de Matemática e Estatística}, \orgname{Universidade Federal Fluminense}, \orgaddress{\street{R. Prof. Marcos Waldemar de Freitas Reis, s/n, Bloco H}, \city{Niterói}, \postcode{24120-201}, %\state{State},
%\country{Brazil}}}

\maketitle

\begin{abstract}
%\abstract{
Real populations are seldom found at the Nash equilibrium strategy. The present work focuses on how population size can be a relevant evolutionary force diverting the population from its expected Nash equilibrium. We introduce the concept of insuperable strategy, a strategy that guarantees that no other player can have a larger payoff than the player that adopts it. We show that this concept is different from the rationality assumption frequently used in game theory and that for small populations the insuperable strategy is the most probable evolutionary outcome for any dynamics that equal game payoff and reproductive fitness. We support our ideas with several examples and numerical simulations. We finally discuss how to extend the concept to multiplayer games, introducing, in a limited way, the concept of game reduction.
%}
\end{abstract}


\section{Introduction}
\label{sec:intro}

Game theory, introduced in~\citet{VonNeumannMorgenstern}, is the branch of mathematics that studies strategic behavior. It has found its way into several applications, notably in economics and biology. See~\citet{Gintis,HofbauerSigmund_1998,Broom_Rychtar_2013} for general works that covers several distinct features of it.

A game, in normal form, is defined by a set of players, a set of elementary strategies available for each player, and a set of (real) payoff functions, one for each player. In this work, we will mostly consider two-player games, and the set of strategies will always be homeomorphic to the finite-dimensional $n-1$ simplex $\Delta^{n-1}\bydef\{\mathbf{x}\in\R_+^n|\sum_{i=1}^nx_i=1\}$. Therefore, the payoff functions will be given by real finite dimensional matrices --- henceforth, we shall write $\R^{m,n}$ to denote the set of all real $m\times n$ matrices. The canonical directions of the simplex are also termed \emph{pure} strategies, while the others are the \emph{mixed} strategies. By extension, we will include pure strategies whenever we say ``mixed'' strategies, except otherwise stated. 

For two-player multi-strategy games, we assume that players \A and \B have $n\in\N$ and $m\in\N$ pure strategies, respectively. The set of accessible/mixed strategies is given by $\Delta^{n-1}$ and $\Delta^{m-1}$, respectively. When player \A plays strategy $\bx\in\Delta^{n-1}$ and player \B plays strategy $\by\in\Delta^{m-1}$, player \A and player \B payoffs will be given by $\bx^\adj\mathbf{A}\by$ and $\by^\adj\mathbf{B}\bx=\bx^\adj\mathbf{B}^\adj\by$, respectively, where $\mathbf{A}\in\R^{n,m}$ and $\mathbf{B}\in\R^{m,n}$ are the payoff matrices for players \A and \B, respectively. Whenever $n=m$ and $\mathbf{A}=\mathbf{B}$, we say that the game is symmetric. 

In the study of human economic behavior, it is commonly assumed that each player tries to maximize his or her payoff. This assumption is so central, that individuals who abide by it are called \emph{rational individuals}. Under this assumption, the concept of \emph{Nash Equilibrium} naturally emerges. It is defined by a set of strategies (for general $N$-player game) in which no individual can unilaterally deviate from the equilibrium and receive a larger payoff. 

However, several empirical studies show that humans frequently deviate from Nash Equilibrium, and consequently, from the rationality assumption itself~\citep{Camerer_1997,mailath1998people}. This can happen due, e.g., to ethical behavior or fairness intention~\citep{Falk_2008,HOFFMAN1994346}, seeking for a better reputation in society~\citep{NowakSigmund_Nature}, heterogeneous interest in some specific individuals in the population (helping some --- possibly family members or people with merit --- to the detriment of others)~\citep{FehrFishbacher_TEJ02,EckelGrossman_GEB96}, the desire to punish what is seen as incorrect (selfish, unethical, etc) behavior~\citep{FalkFehrFishbacher_Econometrica}. One interesting example that makes players deviate from the rational behavior is the so-called \emph{last place aversion}, in which individuals prefer to receive less and not be the last than to receive more and be the last (e.g., salaries in an organization)~\citep{Kuziemko_2014}. From a comprehensive review on this topic, cf.~\cite{Ledyard+1995+111+194}.

From the above discussion, it is clear that the assumption of human rationality explains only part of human behavior. However, instead of considering the deviation of human rationality as an anomaly that requires \emph{ad hoc} explanation, we should ask how and why it evolved and, in parallel, why and how non-rational behavior evolved through biological processes. 

The idea that Darwinian dynamics may favor relative maximization of payoff instead of absolute maximization is not new. It is central to the concept of Hamiltonian spite, in which a decrease of one's payoff is followed by a larger decrease of the competitor's payoff, favoring the evolution of a non-Nash equilibria~\citep{Hamilton_70,Stegeman_04}. Evolution of spite requires finite populations~\citep{Smead_Forber_Evolution_13} and is more likely in small populations~\citep{Forber_Smead_TJP14,DUFWENBERG2000147,Marlowe_eta_PRSB11}. In corporate economics, spite is closely related to the practice of price predation, or dumping. For the purpose of the present work, dumping is defined as the act of abusing a dominant position in the market by selling products below their cost price to drive competitors out of business~\citep{Ordover_1981}. 

The objective of the present work is to use simple and standard game theory models to show that when the number of participants is small, evolutionary dynamics may drive the evolution to a different concept of equilibria. We will introduce the concept of \emph{insuperable} strategy, which may not be rational but guarantees for a sufficiently small population a long-term fixation probability not smaller than any other strategy. In particular, not smaller than the opponent's strategy.

On the other hand, consider a large population in which the interaction between individuals is modeled by a game. Assume further that the population evolves according to a certain dynamic that identifies biological fitness (i.e., reproductive viability) with the game's payoff. It is possible to prove that a strategist adopting the Nash-equilibrium strategy will have a fixation probability not smaller than any other strategy~\citep{ChalubSouza09b,ChalubSouza14a}. 

A clear example of the dichotomy large \emph{vs} small population in terms of the final outcome of the evolutionary dynamics was presented in~\citet{ChalubSouza_2017}. Consider the $N$-player public good game ($N$-PGG), in which any individual has the choice to contribute or not one monetary unity for a common pool. After all contributions, the total value is multiplied by $r>1$ and divided equally among all players, independently of their contribution. If $r<N$ rational players will not contribute, but if $r>N$ the return will be $r/N>1$ for each contributed monetary unit, and therefore, rational players should contribute. On the other hand, for any value of $r$, evolutionary dynamics privileging higher payoff will likely select noncontributors (as his or her payoff will be always higher in one unity of the contributor's payoff). Therefore, for $r>N$, the most likely outcome of the evolutionary dynamics is a (non-Nash) equilibrium of non-rational players.

\begin{remark}
\label{rem:largesmall}
The example of the previous paragraph shows that if the population is smaller than $r$, the most likely result of evolution is non-rational behavior, while for a larger $N$ it is the other way round. The parameter $r$ introduces a natural scale such that we may guarantee the advantage of Nash behavior for large populations. In fact, if the population is sufficiently large and reproductive fitness satisfies the weak selection principle, a player adopting the Nash strategy of the game has a fixation probability always not smaller than the neutral one against any other player. The proof turns out to be more involved than expected, using the subharmonic properties of the solutions to a degenerated partial differential equation (the Kimura equation), cf.~\citet{ChalubSouza09b,ChalubSouza14a}. We are not aware of any simpler proof of this result.   
\end{remark}

The present work will explore the interplay between game theory and evolutionary dynamics for small populations. We will show a powerful result showing that in any normal-form two-player game theory (as defined at the beginning of this work) at least one player has at least one strategy that we call insuperable. Assuming that player \A has one insuperable strategy and adopts it, player \B will have a payoff at most equal to player's \A payoff. Therefore, for a sufficiently small population, the fixation probability of type \A strategist will never be less than the neutral one. In general, the insuperable strategy is not a Nash-equilibrium. For asymmetric games, it is possible that only one player has an insuperable strategy. 

This work is structured as follows: in Sec.~\ref{sec:deff}, we introduce the basic definitions and main results. In Sec.~\ref{sec:examples}, we explore several examples, always highlighting the difference from the Nash equilibrium analysis and showing that frequently the behavior found in nature (including in humans) lies between the Nash equilibrium and the insuperable strategy. In Sec.~\ref{sec:finance}, we show that certain results in finance can be explained similarly. In Sec.~\ref{sec:reduction}, we will show that certain $N$-player games, $N>2$, may be reduced to equivalent (in a sense to be made more precise later on) two-player games and the previous analysis follows.
Finally, in Sec.~\ref{sec:conclusions}, we discuss the implications of the present work to understand several evolutionary phenomena.

\section{Definitions and main results}
\label{sec:deff}

From now on, we consider that vectors are \emph{naturally} column-vectors and denote transposes by ${}^\adj$. We also write $\be_i$ for the $i$-th vector in the canonical base of $\R^n$.

Before stating our results, we introduce some notation:
\begin{definition}
Given a vector $\mathbf{v}$, we say that $\mathbf{v}\ge\mathbf{0}$ if $v_i\ge0$ for all $i$; $\mathbf{v}\gg\mathbf{0}$ if $v_i>0$ for all $i$; $\mathbf{v}>\mathbf{0}$ if $\mathbf{v}\ge \mathbf{0}$, and $\mathbf{v}\neq\mathbf{0}$. Finally, we say that $\mathbf{v}\ge\mathbf{u}$ if $\mathbf{v}-\mathbf{u}\ge\mathbf{0}$ and similar definitions for the other two notations. We also consider similar definitions for $\le$, $<$, and $\ll$.
\end{definition}

Consider a two-player game, in which player \A has $n$ pure strategies and payoff matrix $\mathbf{A}\in\R^{n, m}$ and player \B has $m$ pure strategies and payoff matrix $\mathbf{B}\in\R^{m, n}$. We say that a certain strategy is \emph{insuperable} if a player that adopts it can guarantee a result equal or superior to his opponent. 
More precisely,

\begin{definition}
    Consider a two-player game, possibly non-symmetric, with payoff matrices $\mathbf{A}\in\R^{n,m}$ and $\mathbf{B} \in \R^{m,n}$, for players \A and \B, respectively. For player \A, the strategy $\bx_*\in\Delta^{n-1}$ is insuperable if
    \begin{equation}
    \label{eq:unbeat:A}
    \by^\adj \mathbf{B} \bx_* \leq \bx_*^\adj \mathbf{A} \by\ ,\quad\forall \by \in \Delta^{m-1}\ .
    \end{equation}
    Similarly, for player \B, $\by_*\in\Delta^{m-1}$ is an insuperable if
    \begin{equation}
    \label{eq:unbeat:B}
    \bx^\adj \mathbf{A} \by_* \leq \by_*^\adj \mathbf{B} \bx\ ,\quad \forall \bx \in \Delta^{n-1}\ .
    \end{equation}
If \A and \B have insuperable strategies, $\bx_*$ and $\by_*$, respectively, we say that $(\bx_*,\by_*)$ is an insuperable pair. It is clear, that in this case, both inequalities hold as equalities. If one of those inequalities is strict (i.e., ``$\le$'' is replaced by ``$<$''), we say that the corresponding strategy is \emph{strictly insuperable}.
\end{definition}

For the sake of completeness, we recall that the pair of strategies $(\hat\bx,\hat\by)$ is called a \emph{Nash Equilibrium} if $\hat \bx^\adj\mathbf{A}\hat \by\ge \bx^\adj\mathbf{A}\hat\by$ and $\hat\by^\adj\mathbf{B}\hat\bx\ge \by^\adj\mathbf{B}\hat\bx$, for all strategies $\bx$, $\by$, cf.~\citet{Gintis,HofbauerSigmund_1998}.

\begin{remark}\label{rem:spite}
In~\citet{Nowak:06}, the expression \emph{unbeatable strategy} is used for a strategy $S_k$ such that for all other strategies $S_i$, it is true that $E(S_k, S_k) > E(S_i, S_k)$ and $E(S_k, S_i) > E(S_i, S_i)$, where $E(S_k,S_i)$ is the expected payoff of the strategy $S_k$ against $S_i$. This definition is based on the original ideas of~\citet{Hamilton_1964}, which eventually led to the definition of Evolutionarily Stable Strategy, cf.~\citet{Smith_Price_1973,Smith_1982}. 
Despite the similar name, there are important differences between insuperable strategies and the unbeatable strategies referred to above. First, we do not assume that the game is played in a population, let alone that the population is large. Second, we are concerned with a direct conflict between a certain number of opponents (two, in most of the present work, but check Sec.~\ref{sec:reduction}). Third, in the classical works, the genetic correlation between individuals in the population (that, eventually play against each other) is a fundamental feature of the definition, while in the present case, the relatedness between different players has no role. See also \citet{kojima2006stability}.
\end{remark}

\begin{lemma}
Consider a two-player game with payoff matrices $\mathbf{A}$ and $\mathbf{B}$ for players \A and \B, respectively. Let
 $\mathbf{L}\bydef\mathbf{A}^\adj-\mathbf{B}$. Then $\bx_*$ ($\by_*$) is an insuperable strategy for player \A (\B, respect.) if and only if $\mathbf{L}\bx_*\geq0$ ($\by_*^\adj\mathbf{L}\le \mathbf{0}$, respect.) If $(\bx_*,\by_*)$ is a pair of insuperable strategies, then $\by_*^\adj\mathbf{L}\mathbf{x}_*=0$.  
\end{lemma}

    \begin{proof}
    It is clear that Eq.~\eqref{eq:unbeat:A} is equivalent to $\by^\adj\mathbf{L}\bx_*\ge 0$ for all $\by\in\Delta^{m-1}$. $\Rightarrow$) If $\mathbf{L}\bx_*\not\ge\mathbf{0}$, then there is $i$ such that $\left(\mathbf{L}\bx_*\right)_i<0$, and we take $\by=\be_i$. $\Leftarrow$) Follows from the fact that for all $\by\in\Delta^{m-1}$, $\by\ge\mathbf{0}$. For player \B, proofs follow identically. The last claim is immediate.
\end{proof}

\begin{theorem}
\label{thm:unbeat:A+B}
In a two-player game in normal form, at least one player has an insuperable strategy.
\end{theorem}

Before proving the existence lemma, let us recall, for the convenience of the reader,  a well-known result:

\begin{lemma}[Farkas' Lemma~\protect{\citep[Corolary 36]{Border2020}}]
Let $\mathbf{K}\in\R^{m, n}$, and $\mathbf{b}\in\R^m$. Then, exactly one of the two following statements is true
\begin{enumerate}
\item There exists $\bx\in\R^n$, $\mathbf{x}\ge\mathbf{0}$, such that $\mathbf{Kx}\le\mathbf{b}$.
\item There exists $\by\in\R^m$, $\mathbf{y}>\mathbf{0}$, such that $\by^\adj\mathbf{K}\ge\mathbf{0}$ and $\mathbf{b}^\adj\by<0$.
\end{enumerate}
\end{lemma}

\begin{proof}[Proof of Thm.~\ref{thm:unbeat:A+B}]
We assume that $\mathbf{L}\ne\mathbf{0}$, otherwise the result is immediate.
Let us choose $\mathbf{b}=(-1,-1,\dots,-1)$ and $\mathbf{K}=-\mathbf{L}$. Then, either there is $\bx\geq\mathbf{0}$ such that $-\mathbf{L}\mathbf{x}\le-(1 \dots 1)^\adj$, and therefore $\mathbf{Lx}\gg\mathbf{0}$, or there is $\by>\mathbf{0}$ such that $-\by^\adj\mathbf{L}\ge\mathbf{0}$ and $\sum_iy_i>0$. In the former case, it is clear that $\bx>\mathbf{0}$, otherwise $\mathbf{L}\bx=\mathbf{0}$; therefore $\left(\sum_{i=1}^nx_i\right)^{-1}\bx\in\Delta^{n-1}$ is insuperable; in the latter case, $\left(\sum_{i=1}^my_i\right)^{-1}\by\in\Delta^{m-1}$ is insuperable. 
\end{proof}

For a symmetric game, the following corollary shows that if both players try to have a payoff not smaller than the opponents', an equilibrium will emerge.

\begin{corollary}
For a symmetric two-player game, all players have insuperable strategies. In particular, there is one symmetric pair of insuperable strategies.
\end{corollary}

\begin{proof}
    Follows immediately from Thm.~\ref{thm:unbeat:A+B}.
\end{proof}

In a symmetric game, if both players adopt insuperable strategies, their payoff will be equal. In contrast, in non-symmetric games, it is possible that a strategy guarantees a strictly higher payoff than the opponent, but when this happens, it will happen to only one of the players, as shown in the last result from this section.

\begin{corollary}\label{cor:ubeat:AB}
There is an insuperable pair of strategies if and only if there is no strictly insuperable strategy.
\end{corollary}

This corollary follows from the following version of the Farkas' lemma:

\begin{lemma}[Farkas' lemma, alternative version \protect{\citep[Corollary 44]{Border2020}}]
\label{thm:yaa}
    Let $\mathbf{K}\in \R^{m, n}$. Exactly one of the following alternatives holds: i) Either there exists $\bx\in \R^n$, $\bx\geq\mathbf{0}$ satisfying $\mathbf{K}\bx\gg\mathbf{0}$, or else ii) there exists $\mathbf{y}\in\R^m$, $\mathbf{y}>\mathbf{0}$,  satisfying $\mathbf{y}^\adj\mathbf{K}\leq\mathbf{0}$.
\end{lemma}

\begin{proof}[Proof of Corollary~\ref{cor:ubeat:AB}]
The \emph{only if} part follows from the definition of insuperable and strictly insuperable strategies. The \emph{if} part is proved as follows: assume there are no strictly insuperable strategies. We apply Lemma~\ref{thm:yaa} for player \A and \B, with $\mathbf{K}=\mathbf{L}$ and $\mathbf{K}=\mathbf{L}^\adj$, respectively. By assumption, i) is false, then ii) should be true, i.e., there is an insuperable pair. 
\end{proof}

\section{Examples}
\label{sec:examples}

The public good game (PGG) discussed in the introduction is not in the framework of the previous section, as it is a $N$-player game, with $N$ possibly larger than 2. We will show, however, in Sec.~\ref{sec:reduction}, that the PGG is an example of a game that is equivalent to a two-player game. In the present section, we will be restricted to games originally and naturally defined as two-player.

\begin{example}[The hawk-dove game]
\label{ex:HD}
The hawk-dove game is a classic example in evolutionary game theory in which two individuals dispute a good, with value $G>0$, that cannot be split. There are two possible strategies: to fight until the last consequences, willing to pay the cost of fight $C>G$ (to play \textsc{Hawk}), or to avoid fight (to play \textsc{Dove}). The game is symmetric and the associated payoff matrix is given by $\mathbf{B}=\mathbf{A}=\left(\begin{smallmatrix} \frac{G-C}{2}&G\\ 0& \frac{G}{2}\end{smallmatrix}\right)$. The only symmetric Nash equilibrium (and the only evolutionarily stable state) is given by a mixed equilibrium in which strategy \textsc{Hawk} is played with frequency $G/C\in(0,1)$, either due to a probabilistic behavior of each individual in a homogeneous population, or due to the presence of $G/C$ players of the pure strategy \textsc{Hawk} and $1-G/C$ players of \textsc{Dove} type, or any convex combination of the two previous situations. However, $\mathbf{L}=\left(\begin{smallmatrix}0&-G\\G&0\end{smallmatrix}\right)$. In particular, $\bx=\be_1$ is insuperable for both players. This means that in a small population, the strategy \textsc{Hawk} has a larger than neutral probability of being selected. To understand this simple fact, consider a population of two individuals, disputing one indivisible good in a single iteration. There is no reason to play \textsc{Dove}. Who plays \textsc{Hawk} cannot lose, in the sense that an individual cannot have a lower payoff than the opponent: it will be tied against another \textsc{Hawk} and victorious against a \textsc{Dove}. Note that this does not depend on the cost of the fight $C$, that could be less than the possible gain $G$. In a small population, violence will be selected. See Fig.~\ref{fig:HD} for some numerical experiments.
\end{example}

\begin{figure}
    \centering
\includegraphics[width=0.45\textwidth]{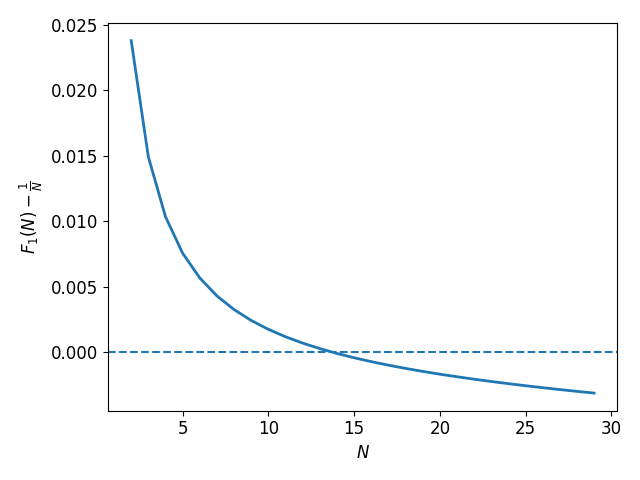}
\includegraphics[width=0.45\textwidth]{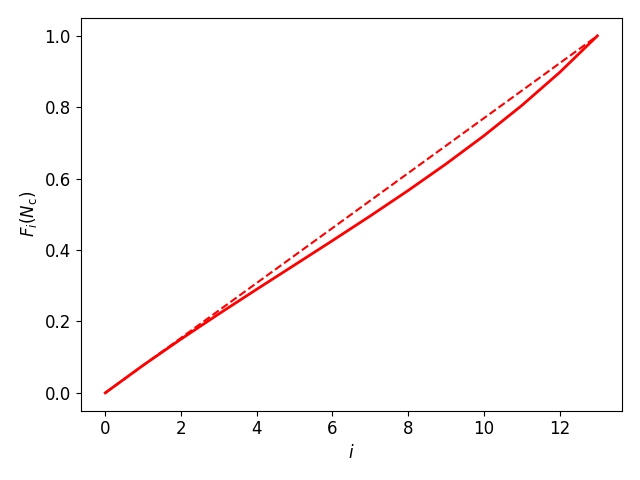}
\caption{We consider a hawk-dove game with $G=3$ and $C=10$, such that the equilibrium point is given by $\frac{3}{10}<\frac{1}{3}$. The last inequality, the so-called $1/3$-law guarantees that natural selection favors \textsc{Dove} to replace \textsc{Hawk} in a large finite  population~\citep{Nowak:06}. We assume the weak selection principle, i.e., payoff matrices are given by $\mathbf{B}=\mathbf{A}=\left(\begin{smallmatrix}1&1\\1&1\end{smallmatrix}\right)+\frac{1}{N}\left(\begin{smallmatrix}\frac{G-C}{2}&G\\ 0&\frac{G}{2}\end{smallmatrix}\right)$. (Left) The invasion probability of a single mutant as a function of the population size compared to the neutral probability $1/N$ as a function of $N$. For small $N$, \textsc{Hawks} have larger than neutral invasion probability and this reverts for $N> N_{\textrm{c}}=13$. (Right) Fixation probability of \textsc{Hawk} as a function of its initial presence $i$ in the case $N=N_{\textsc{c}}$. For $i=1$, the fixation probability (continuous line) is above the neutral (dashed line), however, the difference is of order $10^{-3}$ and is hardly noticeable. For $i\ge 2$, the continuous line is below the neutral one. The explicit expression for $F_1(N)$ is given by Eq.~\eqref{eq:fixacaoMoran}.}
\label{fig:HD}
\end{figure}

\begin{remark}
Evidence on the relationship between aggressiveness and population size is scarce. Frequently, individuals in small populations are closely related, as the reduction in population size may be due to the bottleneck effect and/or to the founder's effect~\citep{Hartle_Clark}, and high population relatedness correlates to lower intraspecific aggressiveness~\citep{Edenbrow_Croft_JFB2012}. This may compensate for the effect described above. In the case of human societies,~\citet{pinker2012better,Knauft_Daly_CA1987} argue that the level of violence in societies organized as modern states is significantly lower than those in nonstate, small-scale societies; cf.~\citet{Falk_Hildebolt} for a different point of view.
\end{remark}

It is not difficult to generalize the above case for any two-player, two-strategy, symmetric game. In this case, $\mathbf{A}=\left(\begin{smallmatrix}a&b\\ c&d\end{smallmatrix}\right)$, and, therefore $\mathbf{L}=\left(\begin{smallmatrix} 0 &c-b\\ b-c & 0\end{smallmatrix}\right)$. If $b>c$, $\be_1$ is insuperable; if $b<c$, $\be_2$ is insuperable. If $b=c$ any strategy is insuperable. More precisely, if $b>c$, a player that adopts strategy 1 can be certain that will not have a payoff smaller than player 2. In particular, in the symmetric two-player game with a payoff matrix given by $\mathbf{B}=\mathbf{A}=\left(\begin{smallmatrix}2&5\\ 4& 8\end{smallmatrix}\right)$, strategy 2 dominates strategy 1. Therefore, the Nash equilibrium is given by $\be_2$, but $\be_1$ is the unique insuperable strategy.

If the population is composed of two individuals, and there is no mutation, we may calculate the fixation probability of type \A individuals using the Moran process. The Moran process is a birth-death process, in which individuals are selected to die with equal probabilities and selected to reproduce (possibly the same one) with probability proportional to the fitness/payoff. Fitness of type \A and type \B individuals are given by $b$ and $c$, respectively. The fixation probability of the type \A individual in a population composed of one type \A individual and one type \B individual is given by weighting the possible paths of this process:
\[
F_\A=\underbrace{\frac{b}{2(b+c)}}_{(1,1)\mapsto (2,0)}\times 1+\underbrace{\frac{1}{2}}_{(1,1)\mapsto(1,1)}\times F_\A+ \underbrace{\frac{c}{2(b+c)}}_{(1,1)\mapsto(0,2)}\times 0\quad\Rightarrow\quad F_\A=\frac{b}{b+c}\ .
\]

We conclude that $F_\A>1/2$, the neutral fixation probability, if and only if $b>c$. This last inequality is equivalent to saying that \A is an $\textsc{ESS}_N$ for $N=2$. Furthermore, for general $N$, $b>c$, and assuming weak selection, the fixation probability of a single \A individual is larger than the neutral one, $1/N$, for $N$ small enough. See~\citet{Nowak:06} for further details.

Finally, in the case of Example~\ref{ex:HD}, $F_\A=1$, and the fixation of \textsc{Hawk} in an initially mixed population of two individuals is certain.

Rmk.~\ref{rem:largesmall} showed that in the PGG there is a sharp transition between a small population $N<r$ and a large population $N>r$. We show that for a general two-strategy two-player game, the dichotomy of small \emph{vs.} large populations is also present, but such a transition is not necessarily sharp.

\begin{theorem}
    Consider a symmetric two-strategy two-player game given by the positive payoff matrix $\mathbf{B}=\mathbf{A}=\left(\begin{smallmatrix} a& b\\ c& d\end{smallmatrix}\right)$ in which the fitness of each individual is equal to the average payoff after playing one game with every other player in the population. Assume a constant size population that evolves according to the frequency dependent Moran process. Assume that \B strictly dominates \A, and that \A is strictly insuperable. Then there are critical population sizes $N_{\mathrm{inf}}$ and $N_{\mathrm{sup}}$, with $N_{\mathrm{inf}}=\min\left\{\frac{d-a}{d-b},\frac{d-a}{c-a}\right\}\le\max\left\{\frac{d-a}{d-b},\frac{d-a}{c-a}\right\}=N_{\mathrm{sup}}$, such that, if $N<N_{\mathrm{inf}}$ ($N>N_{\mathrm{sup}}$), the fixation probability of any initial condition is larger (smaller, respect.) than the neutral one.
\end{theorem}

\begin{proof}
The fixation probability of \A individuals, when their initial presence is given by $i$, is given by $F_0=0$ and, for $i>0$, by:
\begin{equation}\label{eq:fixacaoMoran}
F_i=\frac{\sum_{j=1}^i\prod_{k=0}^{j-1}\rho_k^{-1}}{\sum_{j=1}^{N}\prod_{k=0}^{j-1}\rho_k^{-1}}\ ,
\end{equation}
where the \emph{relative fitness} is given by
\begin{equation}\label{eq:relativefitness}
\rho_k=\frac{a(k-1)+b(N-k)}{ck+d(N-k-1)}\ .
\end{equation}
For a derivation of this expression, see~\citet{KarlinTaylor_First,TaylorFudenbergSasakiNowak,Nowak:06}.
 The relative fitness is the ratio between the average fitness of type \A individuals and type \B individuals after playing one game with every possible pair.

From the dominance by \B and the insuperability of \A, we have that, $d>b>c>a$. Note that 
\[
\rho_k>1\Leftrightarrow \Psi_N(k)\bydef [(d-b)-(c-a)]k-(d-b)N+d-a>0\ .
\]
Assume, initially that $d-b>c-a$; in this case, $N<\frac{d-a}{d-b}$ implies $\Psi_N(0)=d-a-N(d-b)>0$, and, therefore $\Psi_N(k)>0$ for all $k$, i.e., $\rho_k>1$ for all $k$. 

Let
\begin{equation*}
a_i \bydef \prod_{j=0}^{i-1}\rho_j^{-1} ,\qquad
A_i \bydef  \frac{1}{i}\sum_{j=1}^ia_j\ .
\end{equation*}
for $i=1,\dots,N$.
Then $A_i>A_{i+1}$, since it is the arithmetic average of the first $i$ elements of the decreasing sequence $\left(a_i\right)_{i=1,\dots,N}$. Finally, Eq.~\eqref{eq:fixacaoMoran} can then be recast as  
\begin{align*}
F_i = \frac{i}{N}\frac{A_i}{A_N} > \frac{i}{N}.
\end{align*}

If $N>\frac{d-a}{c-a}$, then $\Psi_N(N)=-(c-a)N+d-a<0$ and $\Psi_N(j)<0$ for all $j$. With a similar argument, we conclude that $F_i>\frac{i}{N}$. Note that $d-b>c-a$ implies that $\frac{d-a}{d-b}<\frac{d-a}{c-a}$. The case in which $d-b<c-a$ is similar. If $d-b=c-a$, fitness difference $\Psi_N(k)$ is constant. It is clear that $d-a>d-b$, and for $N$ small (large) $d-a>(d-b)N$ ($d-a<(d-b)N$, respect.). The result follows immediately from the fact that $\rho_k$ is always larger than 1 for small $N$ and smaller than 1 for large N.
 \end{proof}

\begin{remark}
    The difference $N_{\mathrm{sup}}-N_{\mathrm{inf}}$  can be seen as a measure of how sharp is the small \emph{vs.} large transition. In this sense, the PGG is maximally sharp, cf.~Rmk~\ref{rem:largesmall}.
\end{remark}

\begin{example}[Ultimatum game]
    Consider the \emph{ultimatum game}, an example of an asymmetric two-player game. Player
\A (donor) makes an offer, from 0 to $M$ euros (for the sake of simplicity, let us consider only integer offers) to player \B (receiver), which can accept the
offer or not. In the first case, the split is made accordingly; in the second case, both players do not receive anything. We will restrict strategies of player \B to ``$\ge m$'', i.e., ``accept offers greater or equal to $m$, reject otherwise''. The strategy ``$\ge 0$'' means ``always accept'' and ``$\ge M+1$'' means ``always reject''. More complex strategies could be defined but will not be considered here. The strategy $m$ of player \A consists of the offer of $m$ euros. We write the matrices $\mathbf{A}$ and $\mathbf{B}$ for $M=4$ (generalization is straightforward). We use the bi-matrix notation, $(a_{ij},b_{ji})$ to identify player \A and \B payoffs, when \A plays strategy $i$ and \B plays strategy $j$:
\begin{center}
\begin{tabular}{c|cccccc}
\A $\backslash$ \B & $\ge 0$ & $\ge 1$ & $\ge 2$ & $\ge 3$ & $\ge 4$ & $\ge 5$\\
\hline
0 & $(4,0)$ & $(0,0)$ & $(0,0)$ & $(0,0)$ & $(0,0)$ & $(0,0)$ \\
1 & $(3,1)$ & $(3,1)$ & $(0,0)$ & $(0,0)$ & $(0,0)$ & $(0,0)$\\
2 & $(2,2)$ & $(2,2)$ & $(2,2)$ & $(0,0)$ & $(0,0)$ & $(0,0)$ \\
3 & $(1,3)$ & $(1,3)$ & $(1,3)$ & $(1,3)$ & $(0,0)$ & $(0,0)$ \\
4 & $(0,4)$ & $(0,4)$ & $(0,4)$ & $(0,4)$ & $(0,4)$ & $(0,0)$
\end{tabular} 
\end{center}
From the above table, it is clear that all pairs of the form $(m,\ge m)$, $m=0,\dots,M$, $(0,\ge M)$, and $(0,\ge M+1)$ are Nash equilibria.  Furthermore, any pair $(m,\ge m')$, with $m\le \lfloor M/2\rfloor $, and $m'\ge \lceil M/2\rceil $ is insuperable. $\lfloor x \rfloor$ and $\lceil x \rceil$ are the largest integer smaller or equal to $x$ and the smaller integer larger or equal to $x$, respectively. An insuperable pair with a positive money transfer from \A to \B exists only if $M$ is even, and is given by $(M/2,\ge M/2)$. See Fig.~\ref{fig:ultimatum}
\end{example}

\begin{figure}
    \centering
\includegraphics[width=0.32\textwidth]{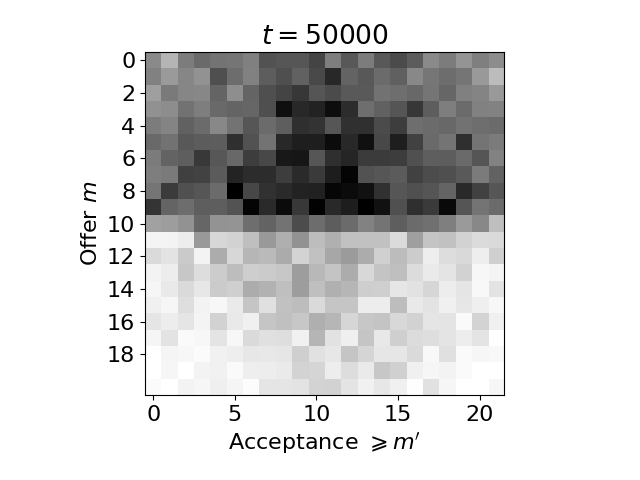}
\includegraphics[width=0.32\textwidth]{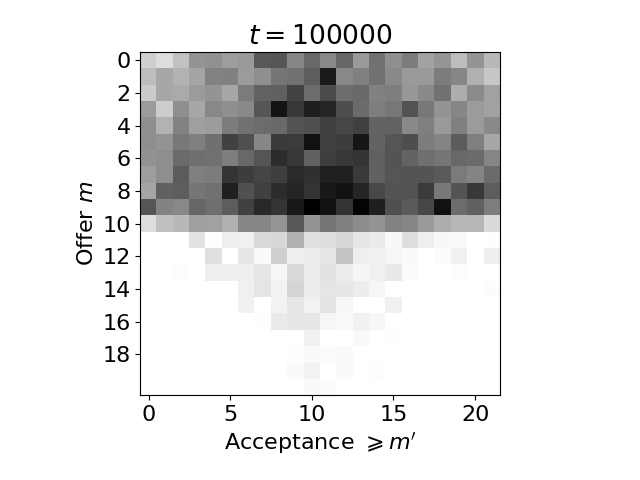}
\includegraphics[width=0.32\textwidth]{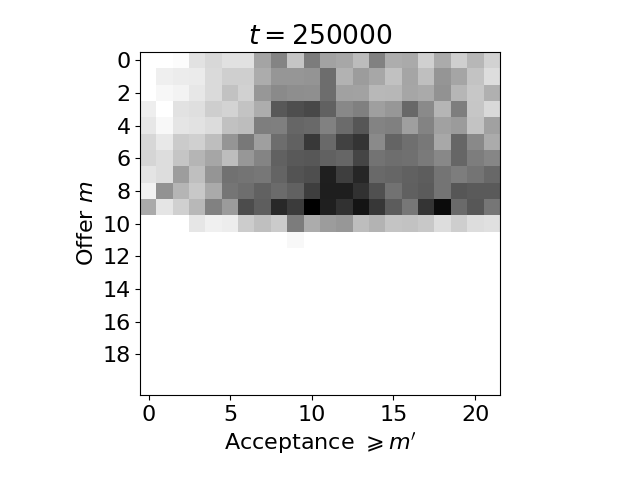}\\
\includegraphics[width=0.32\textwidth]{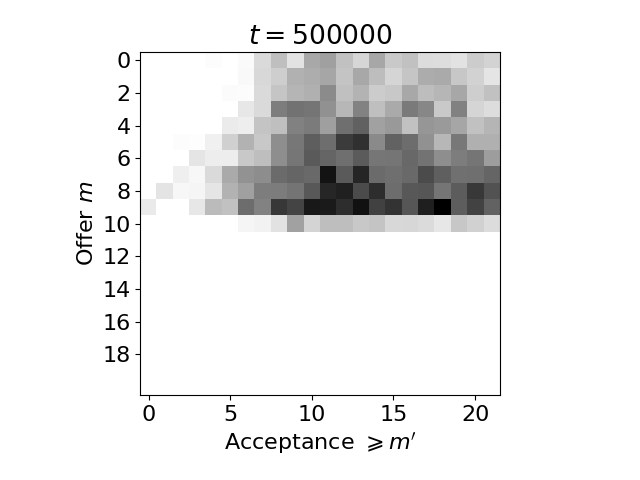}
\includegraphics[width=0.32\textwidth]{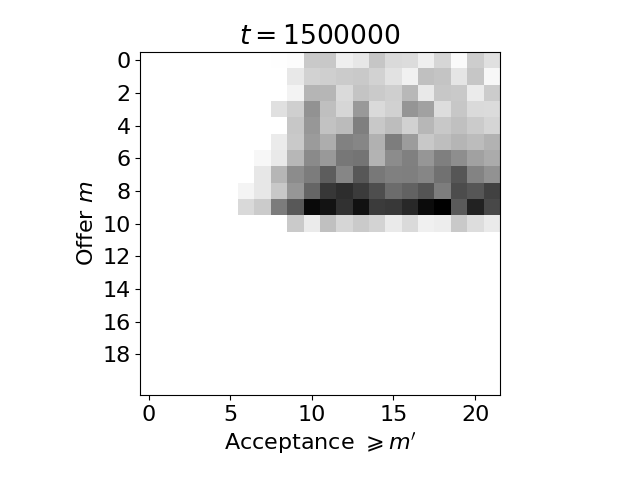}
\includegraphics[width=0.32\textwidth]{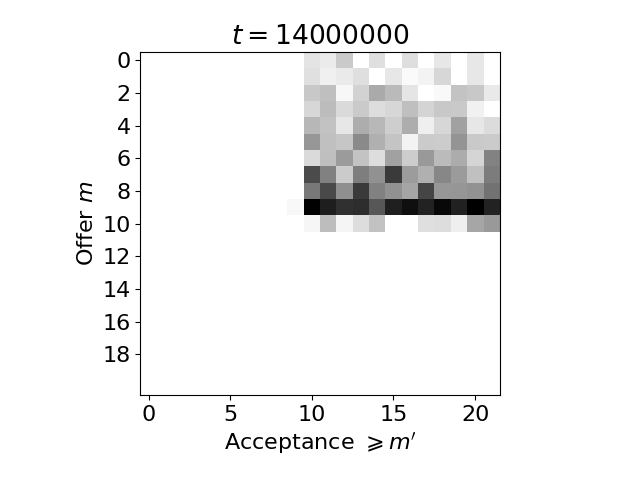}
\caption{Consider a population of $21\times 20\times 100$ individuals initially distributed homogeneously. Each one has a strategy $(m,\ge m')$, for $m=0,\dots,20$, $m'=0\dots,21$. At each time step, two individuals are selected at random to play the Ultimatum game with $M=20$. The one with the higher payoff replaces the one with the smaller payoff. If payoffs are equal, one replaces the other with probability $1/2$. Note that initially all strategists who make high offers are eliminated; after a certain time, the ones who accept lower offers are also eliminated. In the end, only strategists of the form $(m,\ge m')$, $m\le10$, $m'\ge 10$ will survive. Eventually, all encounters among surviving strategists are evolutionarily neutral.}
\label{fig:ultimatum}
\end{figure}

\begin{remark}
 Most experimental results indicate that real players make offers close, but inferior, to the fair share of 50\%~\citep{Croson_1996,NowakPageSigmund,camerer2011behavioral}, i.e., they are between strict Nash equilibrium and the insuperable pair. Experiments in several small-scale societies show a strong connection between the game's result and the player's cultural background. Most of the offers were closer to the insuperable pair described above than to the subgame perfect equilibrium (``to offer the minimum", ``to accept anything") predicted by theorists~\citep{Henrich_etal_2024,Gintis}. 
 \end{remark}

From the above examples, it is clear that, in two-strategy games, insuperable strategies are either a pure strategy or all possible strategies. In the following example, we show a two-player three-strategy symmetric game, in which unique Nash equilibrium is given by a pure strategy, but the unique insuperable strategy is a mixed one.

\begin{example}
Consider a symmetric two-player, three-strategy game with the payoff matrix given by
\[
\mathbf{B}=\mathbf{A}=\left(\begin{matrix}
			1 & 1 & 4 \\
			2 & 1 & 1 \\
			3 & 2 & 5
			\end{matrix}\right)\quad\Rightarrow\quad
\mathbf{L}=\left(\begin{matrix}
			0 & 1 & -1\\
			-1 & 0 & 1 \\
			1 & -1 & 0 \end{matrix}\right)\ .
\]

Strict Nash equilibrium is given by $\be_3$, while an insuperable strategy is given by $\bx_*=(1/3,1/3,1/3)$. Assuming, $\bx_*=(x,y,z)\in\Delta^2$, this last fact follows from 
\[
\mathbf{L}\bx\ge\mathbf{0}\quad\Rightarrow\quad x\ge y\ge z\ge x\quad\Rightarrow\quad x=y=z=\frac{1}{3}\ .
\]

The maximum possible payoff when both players use the same strategy is strictly obtained by the Nash equilibrium strategy $\be_3$. However, $\be_3$ is defeated by $\be_1$, which is defeated by $\be_2$, which is defeated by $\be_3$. Namely $\be_3^\adj\mathbf{A}\be_1=3<4=\be_1^\adj\mathbf{A}\be_3$, $\be_1^\adj\mathbf{A}\be_2=1<2=\be_2^\adj\mathbf{A}\be_1$, and $\be_2^\adj\mathbf{A}\be_3=1<2=\be_3^\adj\mathbf{A}\be_2$. Furthermore, each pure strategy ties with the insuperable strategy, i.e, $\be_1^\adj\mathbf{A}\bx_*=\bx_*^\adj\mathbf{A}\be_1=2$, $\be_2^\adj\mathbf{A}\bx_*=\bx_*^\adj\mathbf{A}\be_2=\frac{4}{3}$, and $\be_3^\adj\mathbf{A}\bx_*=\bx_*^\adj\mathbf{A}\be_3=\frac{10}{3}$.  
Therefore, a Nash-equilibrium player can be outperformed by different strategists, while an insuperable player cannot.

\end{example}

The next example is such that only \B has an insuperable strategy, which, in light of Corollary~\ref{cor:ubeat:AB}, it will be strictly insuperable.

\begin{example}
    Consider the matrix $\mathbf{L}$ given by
    \[
\mathbf{L}=
\begin{pmatrix}
    1&-10\\
    -10&1
\end{pmatrix}\ .
    \]
For any vector $(\alpha,\beta)$, we find
\[
\mathbf{L} \begin{pmatrix}
    \alpha\\\beta
\end{pmatrix}
= \begin{pmatrix} 
\alpha-10\beta\\-10\alpha + \beta 
\end{pmatrix}\ ,
\]
which cannot be non-negative for any choice of $(\alpha,\beta)\geq0$.

On the other hand, 
\[
\begin{pmatrix}
    1&1
\end{pmatrix}
\mathbf{L}=
\begin{pmatrix}
    -9&-9
\end{pmatrix}\ .
\]
Therefore, in a non-symmetric game with such an $\mathbf{L}$ matrix, player \B has an insuperable strategy, but player \A does not.
\end{example}

Finally, we study a classical paradox in economics, introduced in \citet{Selten:78}, which turns out to be an asymmetric $N$-player game. Here, we will consider the case $N=2$. 

\begin{example}[Chain store paradox]
 There is a monopolist (player \A) and a potential competitor (player \B) trying to enter the market; the competitor has two strategies: \textsc{out} (strategy 1) and \textsc{in} (strategy 2), whereas the monopolist has also two strategies: \textsc{cooperate} (\textsc{c}, strategy 1) and \textsc{dispute} (\textsc{d}, strategy 2), should the potential competitor chose \textsc{in}. The payoffs are then
\begin{displaymath}
    \mathbf{A}=\begin{pmatrix}
        5&2\\
        5&0
    \end{pmatrix}
\quad\text{and}\quad
    \mathbf{B}=\begin{pmatrix}
        1&1\\
        2&0
    \end{pmatrix}\ .
\end{displaymath}
The Nash equilibrium for this game is $(\textsc{d},\textsc{out})$, whereas there exists a continuum of insuperable pairs in this game. Indeed, 
\begin{displaymath}
    \mathbf{L}=
    \begin{pmatrix}
        4&4\\
        0&0
    \end{pmatrix}\ .
\end{displaymath}
Hence, any convex combination \textsc{c} and \textsc{d} is an insuperable strategy for \A, while \textsc{in} is an insuperable strategy for \B. Within these pairs, the more profitable is (\textsc{c},\textsc{in}) and this agrees with the analysis coming from subgame perfect equilibrium for a sequential game playing in different (say) cities. This suggests that this choice can be also justified if one assumes that each branch has autonomy in deciding its policy against the local potential competitor, and since locally the game is played only once, the local management chooses the insuperable strategy \citet{Selten:78,rosenthal1981games}; see also \citep{milgrom1982predation} for an alternative view.
\end{example}

\section{A game theoretical view of finance}
\label{sec:finance}

One-period discrete markets --- the so-called Arrow-Debreu model  \citep{duffie2010dynamic} --- and their multi-period extensions are the simplest class of market models used for developments in quantitative finance \citep{pliska1997introduction,follmer2011stochastic}.  In what follows, we will revisit this setup from a game theoretical standpoint --- more specifically, from the perspective of insuperable strategies developed above. Before embarking on this task, we recall the classical setup following~\citet{laurence2017quantitative}.

A one-period market with $m$ assets and $n$ possible states is characterized by a cash flow matrix $\mathbf{D} \in \R^{m,n}$ and by a price vector $\mathbf{p}\in \R^m$. While the latter specifies the price of the $m$ assets at time $t=0$, the former specifies the corresponding cash flows at time $t=1$ depending on the state assumed by the market.
We also write $\btheta\in\R^m$ to denote an allocation of such $m$ assets. More precisely,  if $\theta_i\not=0$ then the investor has taken a position with respect to the $i$-th asset: If $\theta_i>0$, the investor has bought $\theta_i$ unities of the $i$th asset, and we say that the investor is long on this asset; on the other hand, if $\theta_i<0$ the investor has sold $\theta_i$ unities of this asset without owning this asset --- in this case, we say the investor is short on this asset. In particular, $\theta_i<0$ involves the concept of short-selling an asset \citep{hull2006options,luenberger1998investment}.

In order to be useful, a market model must not allow for an allocation that ensures profit at zero cost. In a market without arbitrage, a zero-cost position that yields profit in some states is possible, but it could also lead to losses in other states of the economy. The precise definition under our framework is as follows: 

\begin{definition}[arbitrage strategy]
	An arbitrage is a portfolio $\btheta$ satisfying
	\begin{displaymath}
		\btheta^T \mathbf{D}>0
		\quad\text{and}\quad
		\btheta\cdot\mathbf{p}\leq 0.
	\end{displaymath}
or
	\begin{displaymath}
	\btheta^T \mathbf{D}\geq0
	\quad\text{and}\quad
	\btheta\cdot\mathbf{p}< 0.
\end{displaymath}
\end{definition}

An arbitrage portfolio is then either a portfolio with at most a zero cost, with no downside and a possibility of a positive gain, or a portfolio with a negative cost (the trader makes an immediate profit on assembling such a portfolio) and without a downside. The next theorem characterizes markets without arbitrages (see \citet{laurence2017quantitative} for a proof):

\begin{theorem}[Fundamental theorem of assets]
	There exist no arbitrage portfolios in our one-period market if and only if there exists a positive vector $\bpi\in\R^n$, such that $\mathbf{p}=\mathbf{D}\bpi$; the vector $\bpi$ is termed a state price vector.
\end{theorem}

In our framework, we shall have an additional conical restriction to the market model: any portfolio vector $\btheta$ must satisfy $\btheta\geq0$. 

We are ready to recast the one-period market model above as a non-symmetric game: the cash flow matrix $\mathbf{D}$ will be identified with the net payoff matrix $\mathbf{L}$. Player \A will represent the trader, whereas player \B will represent the market --- a strategy $\bx$ for \A will be a portfolio allocation, while a strategy $\by$ for \B will be a choice of a relaxed state price vector, i.e., $\by>0$ instead of $\by\gg0$ --- with $\mathbf{p}^\adj=-\by^\adj \mathbf{L}$. In this framework, an arbitrage portfolio is such that $L\bx>0$ and $\bx^\adj \mathbf{p}=0$ or  $\textbf{L}\bx\geq0$ and $\bx^\adj \mathbf{p}<0$.

\begin{definition}
    We say that an insuperable strategy has trivial outcome for \A (\B) if $\textbf{L}\bx=\mathbf{0}$ ($\by^\adj \textbf{L}=\mathbf{0}$, respect.). 
\end{definition}

We conclude that section with a necessary and sufficient condition to the existence of arbitrage portfolio, stated in terms of insuperable strategies.

\begin{theorem}
    There exists no arbitrage portfolio in our conical market if and only if every insuperable strategy for the trader has trivial outcome.
\end{theorem}

\begin{proof}
$\Leftarrow$ Assume that every insuperable strategy for the trader has trivial outcome. In this case, the first possibility for an arbitrage is ruled out. Furthermore, if $\textbf{L}\bx=0$, then $\textbf{p}^\adj\bx=0$, which also rules out the second possibility, and thus there is no arbitrage. $\Rightarrow$ Assume there is no arbitrage. In this case, there is no strictly insuperable strategy for the trader --- otherwise, it would be an arbitrage portfolio. Thus the market does have an insuperable strategy. Let $\by$ be that strategy. Assume that the trader has an allocation $\bx$ with $\mathbf{L}\bx>0$, then $\mathbf{y}^\adj \mathbf{L}\bx=0$, and we would have an arbitrage: A portfolio with possibility of profit at zero-cost. Hence for every allocation $\bx$, we must have  $\mathbf{L}\bx\geq0$, which implies either that the trader does not have an insuperable strategy or that its insuperable strategy has a trivial outcome.
\end{proof}

\section{$N$-player games that are reducible to 2-player games}
\label{sec:reduction}

The discussion on public good games in the introduction shows the need to extend the concept of insuperable strategies for symmetric multi-player games. In this section, we will explore two examples of a two-strategy multi-player game that can be systematically reduced to a two-player game, where the discussion in Sec.~\ref{sec:deff} applies. 

The payoff of a two-strategy $N$-player game is defined by  $2N$ parameters, 
\[
\ba^{(N)}=\left(a^{(N)}_k\right)_{k=0,\dots,N-1}\ ,\qquad \bb^{(N)}=\left(b^{(N)}_k\right)_{k=0,\dots,N-1}\ , 
\]
where $a^{(N)}_k$ and $b^{(N)}_k$ are the payoffs of a type \A and a type \B players, respectively, in a $N$-player game against $k$ type \A players and $N-k-1$ type \B players. Strategy \A will be insuperable if $a^{(N)}_k\ge b^{(N)}_{k+1}$ for all possible values of $k$, and \B will be insuperable if $a^{(N)}_k\le b^{(N)}_{k+1}$.

On the other hand, we say that strategy \A dominates strategy \B if $a^{(N)}_k\ge b^{(N)}_k$, and \B dominates \A if $b^{(N)}_k\ge a^{(N)}_k$, in both cases for $k=0,\dots,N-1$.

\begin{remark}
    It is usual to write payoffs of two-strategy $N$-player games in the following table:
    \begin{center}
        \begin{tabular}{c|ccccc}
        & 0-\A & 1-\A & 2-\A & \dots & $(N-1)$-\A\\
        \hline
        \A & $a^{(N)}_0$ & $a^{(N)}_1$ & $a^{(N)}_2$ & \dots & $a^{(N)}_{N-1}$\\
        \B & $b^{(N)}_0$ & $b^{(N)}_1$ & $b^{(N)}_2$ & \dots & $b^{(N)}_{N-1}$
        \end{tabular}\ .
    \end{center}
    In this case, dominance relations are obtained by comparing vertical lines, while insuperability is verified through comparison of $- 45^{\mathrm{o}}$ diagonal lines, i.e. the comparison between $a^{(N)}_k$ and  $b^{(N)}_{k+1}$.
\end{remark}

We introduce the concept of game reduction by saying that the $N$-player game given by $\ba^{(N)}$ and $\bb^{(N)}$ is reducible to a $2$-player game if there are $2$ two-vectors $\ba^{(2)}$ and $\bb^{(2)}$ such that the result of the $N$-player game, for any particular composition of strategists, is equal to the average result a given player would receive, playing separately $N-1$ rounds, each round against a different opponent.

More precisely, the game given by $(\ba^{(N)},\bb^{N})$ is reducible to a two-player game if there are 2 two-vectors $\ba^{(2)}$ and $\bb^{(2)}$ such that 
\begin{align}
\label{eq:reducAtwo}
a^{(N)}_k&=\frac{1}{N-1}\left[ka^{(2)}_{1}+(N-k-1)a^{(2)}_{0}\right]\ ,\\
\label{eq:reducBtwo}
b^{(N)}_k&=\frac{1}{N-1}\left[kb^{(2)}_{1}+(N-k-1)b^{(2)}_{0}\right]\ .
\end{align}
We note that a necessary condition such that the reducibility to a two-player game holds is that the payoff is an affine function with respect to $k$.

For $N=3$, this reduces to
\begin{align}
\label{eq:reducibility0}
    a_0^{(3)}&=a_0^{(2)}\ ,\\
\label{eq:reducibility2}
    a_2^{(3)}&=a_1^{(2)}\ ,\\
\label{eq:reducibility1}
    a_1^{(3)}&=\frac{1}{2}a_0^{(2)}+\frac{1}{2}a_1^{(2)}= \frac{a_0^{(3)}+a_2^{(3)}}{2}\ ,
\end{align}
and the same relations for $\bb^{(3)}$ and $\bb^{(2)}$.

It is clear that if \A dominates \B in the three-player game, then \A will dominate \B in the two-player game and vice-versa. In general, reducibility does not preserve insuperability. However, if $b^{(3)}_2\geq a^{(3)}_2$ and \A is insuperable, then in the reduced game, $b^{(2)}_1\geq a^{(2)}_1$, and, more interestingly, \A is still insuperable:
\begin{equation}\label{eq:reduc_unbeat}
b_2^{(3)}\le a_1^{(3)}=\frac{a_0^{(3)}+a_2^{(3)}}{2}
\quad\Rightarrow\quad a_0^{(2)}=a_0^{(3)}\ge 2b_2^{(3)}-a_2^{(3)}\ge b_2^{(3)}=b_1^{(2)}\ .
\end{equation}
In particular, if \B dominates \A and \A is insuperable in the three-player game, then the same will be true in the reduced game (assuming that it exists).

The converse is not necessarily true, however; i.e., if a two-player game is equivalent to a three-player game, and \B dominates \A in the three-player game, then \B dominates \A in the two-player game. On the other hand, the insuperability of \A in the two-player game will not necessarily be true in the three-player game, but see Rmk.~\ref{rem:extremes}.

The next example will show a well-known multiplayer game, such that \A is insuperable and \B is dominant.

\begin{example}
Consider the $N$-player public good game ($N$-PGG), with multiplicative parameter $r>0$. We consider that \A is a contributor (cooperator) and \B is a non-contributor (defector). Payoffs are given by
    \[
    a^{(N)}_k=\frac{(k+1)r}{N}-1\ ,\qquad 
    b^{(N)}_k=\frac{kr}{N}\ .
    \]
    Notice that $\mathbf{a}^{(2)}=\left(\frac{r}{N}-1,r-1\right)$ and  $\mathbf{b}^{(2)}=\left(0,\frac{(N-1)r}{N}\right)$ satisfy Eqs.~\eqref{eq:reducAtwo} and~\eqref{eq:reducBtwo}, and, therefore, the $N$-PGG is reducible to a series of pairwise interactions. The resulting matrix  is given by
    \[
    \mathbf{B}=\mathbf{A}=\left(\begin{matrix}
        r-1&\frac{r}{N}-1\\
        \frac{(N-1)r}{N}&0
    \end{matrix}\right)\ ,
    \]
and does not represent a payoff in a 2-PGG. 
If $r<N$, then \B dominates and the only Nash equilibrium is $(\B,\B)$ (the tragedy of the commons). For $1<r<N$, the problem is reduced to the prisoner's dilemma, from which the 2-PGG is a particular case. If $r>N$, \A dominates and the Nash equilibrium is given by $(\A,\A)$.  However, for any value of $r>0$, $\frac{r}{N}-1<\frac{(N-1)r}{N}$ and, therefore \B is insuperable: any evolutionary dynamics will likely evolve towards a population of defectors, even considering that this behavior is not rational for $r>N$.
    \end{example}

\begin{remark}
The concept of game reduction can be generalized from $N$ to any $N'<N$, and in this case the affine dependency on $k$, cf. Eqs~\eqref{eq:reducAtwo} and~\eqref{eq:reducBtwo}, may be relaxed. Furthermore, we may consider payoff reduction from $\ba^{(N)}$ to $\ba^{(N')}$, without necessarily considering the same game reduction for $\bb^{(N)}$. When this happens, we are naturally led to the recently introduced \emph{variable size game theory}, in which the number of players involved in a game depends on the strategy of each participant~\citep{HansenChalub_JTB24}. A more general study of game reduction will be done elsewhere.  
\end{remark}
    
The reduction can also be misleading: the next example shows a game without a dominant or insuperable strategy, which is, however, equivalent to a reducible one, but such that the reduced game has insuperable strategies.

\begin{example}
\emph{Zerinho-ou-um} is a child's play that has two winners\footnote{This game, which reminds both authors of their infancy in Rio de Janeiro, Brazil, is, in fact, a pre-game, i.e., a game used to select players for a different, and more important game. We are not aware of any English name for this or a similar game. In Portuguese it reads ``little zero or one'', representing the two possible strategies, and was used to select, for example, among three candidates, the two that would play the first ping-pong match, and who would stand by waiting for the next round.}. In the simplest form, it is a symmetric three-player game, in which all players choose between two strategies: put \textsc{zero} (strategy \A) or \textsc{one} (strategy \B) fingers; if two players make equal moves, then they win, while the other is eliminated. If the three players make the same move, the game is repeated. The payoff table is given by
\begin{center}
    \begin{tabular}{c|ccc}
& 0-\A & 1-\A & 2-\A\\
\hline
\A & 0 & 2 & 1\\
\B & 1 & 2 & 0
    \end{tabular}\ ,
\end{center}
where 0 means defeat (playing \A against two \B, or the other way round), 2 means victory (playing \A or \B against one \A player and one \B player), and 1 means that a new round must be played (all players make the same move).

Neither pure strategy is dominant nor insuperable. There are three Nash equilibria in this game: $(\A,\A,\A)$, $(\B,\B,\B)$, and the symmetric equilibrium with mixed strategy $x=\frac{1}{2}$, that consists in playing \A or \B with equal probability. (This is one of the weakest points of playing this game: if all players start, by chance, playing the same strategies, it is possible to enter a very long series of repeated rounds without a winner, as no one would be available to change strategy. It would pay to change strategy, only if one is suspicious that any other player is likely to change; this is strikingly different from the more well-known Rock-Scissor-Paper game, in which after a tie there is a strong incentive to be the first to change strategy.)

The game defined in the previous table is not reducible, but as we will show, an immaterial change in the payoff matrix will make it formally reducible. We may understand the \emph{zerinho-ou-um} as a series of three pairwise comparisons, in which both players receive a certain amount of points $\alpha>0$ if they match and receive zero otherwise. In this case, if all players opt for the same strategy, they will receive an average of $\alpha$ points. If two players opt for the same strategy (say, \A) and one for the other (say, \B), then the \A players will receive $\alpha/2$, and the \B player will receive 0. The new game payoff table is    
\begin{center}
    \begin{tabular}{c|ccc}
& 0-\A & 1-\A & 2-\A\\
\hline
\A & 0 & $\alpha/2$ & $\alpha$\\
\B & $\alpha$ & $\alpha/2$ & 0
    \end{tabular}\ .
\end{center}
The interpretation of the above table is not so straightforward, but in a game in which we compare the difference of points among players, $a^{(N)}_{N-1}$ and $b^{(N)}_0$ are immaterial, as they related to the payoff of an \A (\B, respect.) player in a homogeneous population of \A (\B, respect.) players. 

The new game is reducible and the payoff matrix of the reduced game is given by $\left(\begin{smallmatrix}\alpha&0\\0&\alpha\end{smallmatrix}\right)$, a coordination game. 
Nash equilibria are given by $(\A,\A)$, $(\B,\B)$ and strategy $x=\frac{1}{2}$, namely to play \A or \B at random, with probability $1/2$. On the other hand, in the two-player reduction, any strategy is insuperable. This follows from the fact that $\mathbf{L}=\mathbf{0}$. In the end, there is no reason to play a particular strategy in this game, if we are not expecting a long chain of interactions that could allow the detection of the opponents' strategies.
Extension of the reduced two-player game to a $N$-player game is straightforward: $a^{(N)}_k=k\alpha$, $b^{(N)}_k=(N-1-k)\alpha$. In particular, after one round of the \emph{zerinho-ou-um} game, all players that used the minority strategy are eliminated. If one of the groups is empty or if there is a tie, the game is repeated.
\end{example}

\begin{remark}\label{rem:extremes}
    The change to the \emph{zerinho-ou-um} game from irreducible to reducible is more general than it seems. In fact, when we are interested in a direct comparison of game payoffs, the value of $a^{(N)}_{N-1}$ and $b^{(N)}_0$ are arbitrary, as they report the game result in a homogeneous population. Therefore, any three-player game can be made reducible, and any two-player game in which \B dominates and \A is insuperable may be changed to be equivalent to a three-player game in which the same is true. It is enough to change the value of $a^{(2)}_1$ (keeping it smaller than $b^{(2)}_1$) and $b^{(2)}_0$ (keeping it larger than $a^{(2)}_0$) such that from $a^{(2)}_0\ge b^{(2)}_0$ it is possible to conclude that 
    \[
a^{(3)}_0=a^{(2)}_0\ge \frac{b^{(2)}_0+b^{(2)}_1}{2}=b^{(3)}_1\ ,\quad
a^{(3)}_1=\frac{a^{(2)}_0+a^{(2)}_1}{2}\ge b^{(2)}_1=b^{(3)}_2\ .
    \]
    Although these modified games are equivalent from the point of view of direct payoff comparisons, their payoffs are different in absolute value. Therefore inferences such as that an insuperable strategy is the most likely strategy to be fixated in small populations will not necessarily hold. We also point out that the strong dependence of the game dynamics in finite populations, outside the weak-selection regime was already found in \citet{chalub2016fixation}.
\end{remark}

\section{Discussion and conclusions}
\label{sec:conclusions}

The concept of \emph{insuperable strategy}, introduced here, is useful when we consider players that try to outperform their competitors. Instead, when one assumes that they are rational, as economic theory normally does, Nash equilibrium strategy is the expected outcome of a normal-form game.

However, in many situations, we directly compete with an opponent, and the primary objective of the participant is to win first place, even if this is associated with minimum victories. Think about sports competitions, games such as the war of attrition~\citep{Broom_Rychtar_2013,BishopCannings_1978}, or, if this sounds far-reaching, think of direct competition between economic agents when it is clear that the number of firms is larger than the market capacity~\citep{BulowKlemperer}. Some of them will go bankrupt, and more than short-term profit maximization, an important drive will be short-term survival. The same applies to a small population if every interaction has life-or-death consequences. By continuity, we cannot expect that Nash equilibrium emerges as a consequence of evolutionary dynamics in small populations, identifying reproductive capacity (or fitness) with game payoffs.

In most of human existence, we lived in communities that are small when compared to the ones that we have in modern life. Therefore, it is possible that we have hard-wired in our genetic background behaviors that seem maladaptive in the current society but that could be justified as a memory of the past with a smaller number of social interactions, cf.~\citet{Rayner_Sturiale} for examples of contemporary maladaptive behavior that were positively selected in different conditions that existed in the past. Two possible examples in the case of human behavior are the case of \emph{envy}, the pain caused by the
sight of other people's good fortune~\citep{Benzeev_1992}, and \emph{spite}, the act of paying (losing payoff) to make one's direct competitor lose even more~\citep{Hamilton_70}. 

In the Ultimatum game, strategies $\ge m'$, with $m'>1$ indicates \emph{spiteful} behavior, as the recipient is willing to receive nothing instead of $0<m''<m'$, to punish the donor~\citep{Marlowe_eta_PRSB11}; in this sense, our model shows clearly, at least in one example, the evolution of spiteful behavior. Still in the Ultimatum game, if we consider that despite the donor's and recipient's strategies evolving independently, there is a strong selective pressure to have $m'=m$, then the only symmetric insuperable strategy will be given by $(\frac{M}{2},\ge\frac{M}{2})$ (assuming $M$ even). This selective pressure can be caused by social pressure/ethical behavior, well expressed in the quote: \emph{It is not fair to ask of others what you are not willing to do yourself}~\citep{Roosevelt}.

Although the evidence is not strong, previous studies show that young children tend to maximize payoff, even in games specifically designed to privilege relative maximization of payoff, while older children understand better situations in which payoff maximization is preferred (playing the Nash equilibrium strategy), and which payoff \emph{relativization} is preferred, playing the insuperable strategy~\citep{Toda_1978,d_almeida_2015_19607}.

One important feature of the present work, which differs from several studies in human behavior (cf. the discussion in~\citet{Camerer_1997}), is to consider deviation of the Nash equilibrium as the expected behavior of game participants with different goals, rather than as anomalies or maladaptive behavior. 

Our work was partially inspired by the concept of Hamiltonian spite~\citep{Hamilton_70}. However, we are not aware of any mathematical study of this concept with the depth and generalization provided here. Furthermore, many technical details in the present work are different from the study of the spite dynamics, as stressed in Rmk.~\ref{rem:spite}. Another source of inspiration for the present work is the concept of \emph{Evolutionarily Stable Strategy} in finite populations, namely, in a population of two individuals playing two pure strategies: this is the $\textsc{ess}_2$, introduced in~\citet{NowakSasakiTaylorFudenberg}; see also the discussion in~\citet{Nowak:06}. Previous works, however, do not extend this concept to non-symmetric or to multi-strategy games, as we did here, or even to multi-player games, as we also did, although in a limited scope.

From a dynamic viewpoint, the concept of insuperable strategy stems from an apparent contradiction that is seldom explored: The Nash equilibrium is not necessarily the attractor of finite population evolutionary dynamics. With this rationale in mind, we introduced a new concept of equilibrium: The insuperable strategy, initially in the domain of two-player two-strategy games and subsequently extended to multi-player two-strategy symmetric games. Furthermore, this new notion of equilibrium leads us to the study of reducible games, or, equivalently, multiplayer extensions of games. The latter prompts us to study how concepts such as dominance, and insuperability propagate from games with a small number of players to games with a large number of players and vice-versa. This is an ongoing work.

\section*{Acknowledgements}
All the authors are funded by Portuguese national funds through the FCT – Fundação para a Ciência e a Tecnologia, I.P. (Portugal), under the scope of the projects UIDB/00297/2020 (https://doi.org/10.54499/UIDB/00297/2020) and UIDP/00297/2020 (https://doi.org/10.54499/UIDP/00297/2020) (Center for Mathematics and Applications --- NOVA Math). FACCC also acknowledges the support of the project \emph{Mathematical Modelling of Multi-scale Control Systems: applications to human diseases}  2022.03091.PTDC (https://doi.org/10.54499/2022.03091.PTDC), supported by national funds (OE), through FCT/MCTES. (CoSysM3). Part of this work was done during the stay of FACCC at Carnegie Mellon University, supported by the CMU-Portugal Program. Furthermore, part of this work was done during FACC stays at City University, London (UK), Universidade Federal do Ceará (Brazil), and Instituto de Matemática Pura e Aplicada (Rio de Janeiro, Brazil). FACCC also acknowledges discussions on the concept of Hamiltonian spite with André d'Almeida, which eventually led to preliminary ideas in the concept of insuperable strategies. 
MOS also acknowledges the support of CAPES/BR - Finance code 01 and FAPERJ through grant E-26/210.440.2019. All authors acknowledge the support of the CAPES PRINT program at UFF through grant 88881.310210/2018-01.
Last, but not least, FACCC acknowledges the input of his daughter, Alice, to explain how the $N$-player generalization of the \emph{zerinho-ou-um} game is played nowadays by school children.

\end{document}